\documentclass[lettersize,journal]{IEEEtran}
\usepackage{amsmath,amsfonts,amssymb}
\usepackage{algorithm}
\usepackage{algpseudocode}
\usepackage{xcolor}
\def\BibTeX{{\rm B\kern-.05em{\sc i\kern-.025em b}\kern-.08em
    T\kern-.1667em\lower.7ex\hbox{E}\kern-.125emX}}
\usepackage{amsthm}
\usepackage{pifont}
\usepackage{array}
\usepackage[normalem]{ulem}
\usepackage{caption}
\usepackage{subcaption}
\usepackage{textcomp}
\usepackage{stfloats}
\usepackage{url}
\usepackage{verbatim}
\usepackage{graphicx}
\usepackage{cite}
\usepackage{url}
\usepackage{tablefootnote}
\usepackage{booktabs}
\algnewcommand{\LineComment}[1]{\State \(\triangleright\) #1}

\theoremstyle{definition}
\newtheorem{definition}{Definition}[section]
\newtheorem*{remark}{Remark}
\begin{document}

\title{Correctness of Flow Migration Across Network Function Instances}
\author{\IEEEauthorblockN{Ranjan Patowary\IEEEauthorrefmark{1},
Gautam Barua\IEEEauthorrefmark{2}, Radhika Sukapuram\IEEEauthorrefmark{3}}

\IEEEauthorblockA{
\IEEEauthorrefmark{1} Central Institute of Technology Kokrajhar, India,\IEEEauthorrefmark{2}\IEEEauthorrefmark{3} Indian Institute of Information Technology Guwahati, India \\
Email: \IEEEauthorrefmark{1}r.patowary@cit.ac.in,
\IEEEauthorrefmark{2}gb@iiitg.ac.in,
\IEEEauthorrefmark{3}radhika@iiitg.ac.in}
\thanks{This work has been submitted to the IEEE for possible publication. Copyright may be transferred without notice, after which this version may no longer be accessible.}
}

\maketitle
\IEEEpubid{0000--0000/00\$00.00~\copyright~2021 IEEE}
 
\IEEEpubidadjcol 

\begin{abstract}
Network Functions (NFs) improve the safety and efficiency of networks. Flows traversing NFs may need to be migrated to balance load, conserve energy, etc. When NFs are stateful, the information stored on the NF per flow must be migrated before the flows are migrated, to avoid problems of consistency. We examine what it means to correctly migrate flows from a stateful NF instance. We define the property of Weak-O, where only the state information required for packets to be correctly forwarded is migrated first, while the remaining states are eventually migrated. Weak-O can be preserved without buffering or dropping packets, unlike existing algorithms. We propose an algorithm that preserves Weak-O and prove its correctness. Even though this may cause packet re-ordering, we experimentally demonstrate that the goodputs with and without migration are comparable when the old and new paths have the same delays and bandwidths, or when the new path has larger bandwidth or at most 5 times longer delays, thus making this practical, contrary to what was thought before. We also prove that no criterion stronger than Weak-O can be preserved in a flow migration system that requires no buffering or dropping of packets and eventually synchronizes its states.
\end{abstract}


\section{Introduction}
\label{intro}
%
%
%
%
\IEEEPARstart{M}{iddleboxes} used to improve the safety and efficiency of a network and functionalities of wireless networks such as mobility management are implemented in
software, on commodity hardware, as Network Functions (NFs) \cite{yi2018comprehensive}. This reduces expenses. 
Flows may be migrated from a set of 
NFs (the old path) to another set of
NFs (the new path)  to reduce load on a set
of network nodes, to consolidate flows  to
conserve energy, etc. 

A \textit{flow} is a sequence of packets that have the same set of headers.
The set of values associated with various structures or objects of a node
related to a flow is called a \textit{state}.  If the old path consists of
nodes that do not maintain state information related to flows, flow migration
only requires updating the rules on the old and the new network nodes so that
the flow can change course in a per-packet consistent (PPC) manner --- that is, every packet of a flow traverses either the old path or the new path and never a combination of both, thus avoiding packet drops or loops
\cite{sukapuram2019ppcu}. However, if a subset of the nodes on the old path
maintain state information with respect to flows, migrating flows would
\textit{additionally} require state information to be migrated to the
corresponding set of nodes on the new path before packets are forwarded to the
new path. 

Programmable data planes 
 are used to implement   applications such as load
balancing 
and end-host functionalities such as
consensus control 
\cite{hauser2023survey, ZHANG2021107597}. 
Various solutions exist on flow migrations from NFs implemented on servers \cite{gember2015opennf,wang2017consistent, gember2015improving,
szalay2019industrial, wang2016transparent} and from NFs or other applications implemented on programmable data planes \cite{luo2017swing, xing2020secure,he2018p4nfv}. If states are not migrated correctly, instances of the new stateful nodes may function
incorrectly, resulting in dropped packets or anomalous behaviour of nodes,
often compromising network safety.
Our focus is to formally define what constitutes a
correct flow migration, if network nodes are stateful in general, and in particular,
if they maintain per-flow states, assuming that the rules on paths are updated without drops, loops or congestion, to define a ``weak'' property that must be preserved during migration, and to demonstrate that preserving the weak property requires no buffering or packet drops, and yet, is sufficient in practice. Buffering is avoided to mitigate the  risk of buffer overflow 
\cite{sukapuram2021loss}.

\IEEEpubidadjcol 
Consider the network in Fig. \ref{prop} that depicts migration of a flow $fl$
from an instance of an NF, $NF_{1a}$ (\textit{source NF} (sNF)), to another instance of the same NF,
$NF_{1b}$ (\textit{destination NF} (dNF)).  Both the NF instances and the functions required to migrate a flow
comprise the Flow Management System (FMS).  Let the NF under
discussion be a Network Address Translator (NAT).  After the first packet of
$fl$, a SYN, reaches $NF_{1a}$, a tuple comprising a source IP address and a source
port, denoted as $y_1$, is chosen corresponding to its original source IP address and source port, 
denoted as $x_1$.
The chosen tuple is stored in $NF_{1a}$ corresponding to $fl$. When $fl$ is
subsequently migrated, 
the first packet $p_1$ that arrives at
$NF_{1b}$ needs to wait for the state $state[x_1,y_1]$ corresponding to it to arrive at $NF_{1b}$ from $NF_{1a}$, in
order that it can be forwarded. If this information is not available at $NF_{1b}$, $p_1$ must be dropped. 

\begin{figure}[t]
\centerline{\includegraphics [scale=0.4] {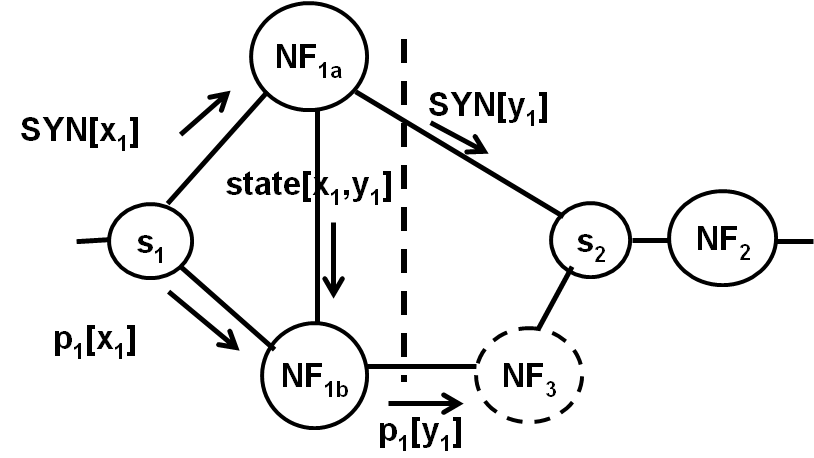}}
\caption{Properties of flow migration}
\label{prop}
\end{figure}
\IEEEpubidadjcol 
If the NF instances $NF_{1a}$ and $NF_{1b}$ in Fig. \ref{prop} correspond to a
Network Intrusion Prevention System (NIPS), they require that packets arrive at
the NF instances in the order that they enter the FMS. 
Moreover, if $p_1$
reaches $NF_{1a}$ and $p_2$ reaches $NF_{1b}$, the state corresponding to $p_1$
at $NF_{1a}$ must be updated at $NF_{1b}$ before $NF_{1b}$ processes $p_2$, for correct operation. 

Previous work~\cite{sukapuram2021loss} defines the property of Order (O), which is a correctness
criterion for flow migration. Informally, if an FMS preserves Order, the
sequence of states created in the source and destination NFs are subsequences
of the states created in the original node if there is no migration (called the
\textit{ideal NF}).  This implies that every state created in
the sNF during flow migration needs to be migrated immediately to the
dNF and vice-versa, in real-time order. 

In Fig. \ref{prop}, assume that $NF_{1a}$ and $NF_{1b}$ are NATs, $NF_2$ is a NIPS,
and $NF_3$ is not present, for now. 
For correct functioning of $NF_2$, it is required that packets enter
$NF_2$ without being re-ordered,
as though no flow migration is occurring.
If it is External-Order (E)~\cite{sukapuram2021loss}
that is preserved, the order in which packets exit the flow migration system (at the interface shown by the vertical dotted line) is
the same as the order in which they exit an ideal NF if there was no
migration. Thus O ensures that the states in the FMS
are exactly the same and in the same order as an ideal NF, while 
E ensures that the packets that exit the FMS
are exactly the same and in the same order as an ideal NF. If a flow migration
preserves E, subsequent NFs will not be aware that a flow migration has
taken place. 
There may be other impacts if packets are reordered: some NFs attempt preventing re-ordering of packets by buffering packets that arrive out of order \cite{yu20163}. Re-ordering would thus cause increase in buffer
sizes in such NFs. We examine these constraints in Section \ref{discussion}.

Assume that  a
NIPS $NF_3$ is inserted after $NF_{1b}$ \cite{zave2020verified} in Fig.\ref{prop}. 
After the state created on $NF_{1a}$ for a SYN is migrated to $NF_{1b}$, if the
packets entering $NF_{1b}$ are re-ordered, 
 the working of $NF_3$ is affected.  If a flow migration
preserves Strict-Order (SO) \cite{sukapuram2021loss}, in addition to preserving O, the
FMS will ensure that packets maintain the order in which they
enter the system while they enter the source or the destination NFs, as the
case may be, whether the packets effect any change in the state of these nodes
or not. 

Ideally, every flow migration should
preserve the highest of the above properties in the hierarchy, that is, E \cite{sukapuram2021loss}, in
order that subsequent stateful nodes are unaware of the occurrence of a flow
migration at nodes preceding it in the flow.
However, preserving E is difficult to achieve, as it requires buffering packets
either before migration or after migration. Similarly, guaranteeing to preserve
O or SO requires that packets are either buffered (thus not preserving the No-buffering property (N)) or dropped (thus not preserving the Loss-freedom property(L)), as per the LON theorem~\cite{sukapuram2021loss}.  In
this paper, \textit{ we  explore the correctness criteria for flow migration and answer these questions:
Can there be weaker properties that are practical for flow migration from one NF instance to another? How can such properties be preserved during flow migration? What is the strongest property that can be preserved while not buffering and not dropping packets?}

Our contributions are: 
1) We argue that a flow can be divided into parts based on whether NF states need to be
synchronized before further packets can be forwarded or not. Based on this, we propose the consistency property of \textit{Weak-O}, that weakens O. 
2) For the first time, we present an algorithm that migrates flows preserving Weak-O and without buffering or dropping packets 
 3) We implement the algorithm and demonstrate that  
a) the goodput with and without migration
are comparable when the old and new paths have the same
delays and bandwidths, or when the new path has larger
bandwidth or at most 5 times longer delays, for TCP flows. b) The time for migration does not matter as no buffers are maintained or packets are dropped. c) Due to the advanced congestion control algorithms that can detect reordering and reduce retransmissions  \cite{arianfar2012tcp}, reordering is not as much an issue as was thought in the seminal work on flow migration \cite{gember2015opennf} and subsequent literature. The above were published as a conference paper in ICDCN 2024. In this paper, we prove the correctness of the algorithm (Theorem \ref{theorem:preserves}, Section \ref{section:WeakO}). We also prove that no criterion
stronger than Weak-O can be preserved in a flow migration
system that requires no buffering or dropping of packets and
eventually synchronizes its states (Section \ref{section:strongest}).
\section{Model}
\label{sys-model}
The model and definitions of O, SO and E are quoted from our previous work \cite{sukapuram2021loss} (with minor edits) for clarity of exposition:

\begin{figure}
\centerline{\includegraphics [scale=0.4, trim=30 20 10 0] {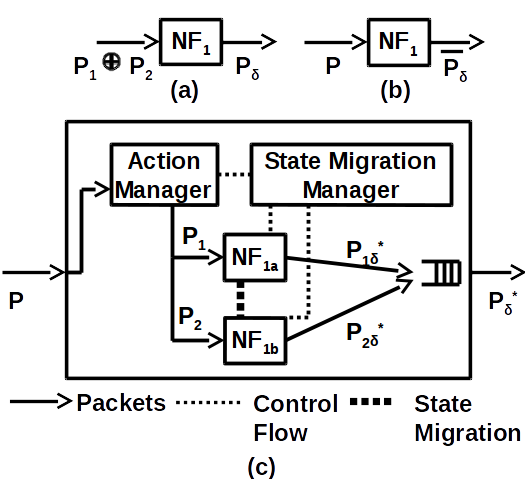}}
\caption{System Model for NF State Migration}
\label{LON-main}
\end{figure}
``States that are updated exclusively by a flow are discussed in
this paper and common states updated by more than one flow (for example,
counters) are outside the scope of discussion. If a set of flows are migrated together
because they share common states, such as proxies that consolidate multiple TCP connections 
into one connection, that is within the scope of discussion.

During migration of a flow, an NF is a tuple $T = (Q, P, P_{\delta}, q_0, 
f, M, U)$ where: 
1) $Q$ is the set of states of the NF for a particular flow (Note: a state is a set of values). 
2) $q_0$ is the initial state, that is, the state of the NF before migration begins
3) $P$ is a set of
input packets 
4) $P_{\delta}$ is a set of output packets, including a $\triangledown$, indicating no packet (explained ahead).
5) $M$ is a set of messages of type $m$, including a $\bot$, indicating no message (explained ahead). 
$m$ contains the state of the NF (source or destination) 
sending the message and the identifier 
of the receiving NF.
6) $f: Q \times P \rightarrow (P_{\delta} \times Q \times M)$
is the transition function. 
7) $U: M \times Q \rightarrow Q $ is the function that receives the message sent by another NF and updates
the current NF. 

\begin{table}
\caption{List of symbols used}
\begin{center}
\begin{tabular}{p{0.2\linewidth} p{0.7\linewidth}}
\toprule
\textbf{Symbol} & \textbf{Meaning} \\
\midrule
$fl$ & The flow being migrated \\
$NF_{1a}$ ($NF_{1b}$) & Source (Destination) NF instance \\
$NF_1$ & An ideal NF \\
$P$ & The sequence of packets entering a flow migration system \\
$P_1$ ($P_2$) & The sequence of packets entering the source (destination) NF \\
$P_1 \oplus P_2$ & The sequence
of packets entering the source and destination NFs, in the order of 
timestamps \\
$p_i$ & A packet belonging to $P$ \\
$s_i$ & A sequence of output packets corresponding to an input packet\\
$q_i$ & An NF state \\
$m(k,q_i)$  & Message sent to NF instance $k$ with state $q_i$ \\
$P_{\delta}$ & Output of an ideal NF when the input is $P_1 \oplus P_2$\\
$\overline{P_{\delta}}$ & Output of an ideal NF when the input is $P$ \\
$Q_1$ ($Q_2$) & The finite sequence of states created on the source (destination) NF\\
$Q_{NF_1}$ & The finite sequence of states created on the ideal NF when its input is $P_1 \oplus P_2$ \\
$Q_{1}^{*}$ ($Q_{2}^{*}$) & A subset of $Q_1$ ($Q_2$) \\
$P_{1 \delta}^{*}$ ($P_{2 \delta}^{*}$) & The sequence of packets output by the source (destination) NF \\
$P_{\delta}^{*}$ & The sequence of packets output by the FMS, in order of being output: $P_{1 \delta}^{*} \boxplus P_{2 \delta}^{*}$\\ 
$\widehat{X}$ & A set representing the elements of $X$\\
$\approxeq_{SS}$ & Partial equivalence of two states with respect to a subsequence of substates denoted by $Q_{ss}$\\
$\models$ ($\not \models$ ) & Immediate synchronization is (not) required\\
\bottomrule
\end{tabular}
\label{symbols}
\end{center}
\end{table}

Further symbols  are listed in Table \ref{symbols}. Let $fl$ be the flow that is being migrated. Fig. \ref{LON-main} depicts
migration of this flow from $NF_{1a}$, an instance of  an NF called $NF_1$, to
another instance of the same NF, $NF_{1b}$. Let us assume that $NF_1$ 
(Fig. \ref{LON-main} (a)), $NF_{1a}$ and
$NF_{1b}$ have the same starting state, $q_0$. The system for NF migration,
referred to as the Flow Migration System (FMS), consists of an Action Manager
(AM), a State Migration Manager (SMM) and the source ($NF_{1a}$) and
destination ($NF_{1b}$) NFs. What is described is a logical organization of FMS
--- in the general case, the source and destination NFs are on different
servers connected over a network. AM manages the following actions on packets:
drop, buffer and nop (indicating no operation). AM also accords a logical
timestamp to each packet when it is received \footnote{This is not a requirement on any flow migration algorithm and is only for purposes of explanation}. 
SMM
decides the NF instance to which a packet $p_i$ of the flow $fl$ must be sent
to, and manages migration of state, if any, from the source to the destination
NF. To aid this, it may instruct AM to drop or buffer packets or forward
packets to an NF instance, using messages
to AM, NFs, both, and/or other network nodes. AM will forward a
packet without delay, unless it is asked to buffer or drop it. SMM may also
instruct NFs to start and stop updating each other's states. AM and SMM may be
co-located with NFs, switches or the controller, or may be independent
entities.

Let $\prec$ be a total order on the \textit{sequence} of packets $P= \langle
p_1, p_2, ..., p_n \rangle$, belonging to $fl$, where $\prec$ denotes the order
in which packets enter FMS.  As packets
$\langle p_1, p_2, ..., p_n \rangle$ enter the NF, they are transformed to 
$\langle s_1, s_2, ..., s_n \rangle$ (denoted as $P_{\delta}$) 
, depending on the state at the NF. 
Transformation of packets is by a function $f()$, where $f(p_i,
q_{i-1})= \langle s_i, q_{i}, m(k,q_{i}) \rangle$, $1 \leq i \leq n$,
where $q_i$ denotes the state of the NF.
$s_i = \langle p_{\delta_{i1}}, ..., p_{\delta_{ic}} \rangle$  denotes a sequence of zero, one or more output packets.  An output packet may be obtained by altering the input packet
(by changing its header, for example).  If there is no change, and only one 
packet is output, $s_i = \langle p_{\delta_{i1}} \rangle =
\langle p_i \rangle$. It is possible that an NF drops an input packet, in which
case no packet ($\langle \triangledown \rangle$) will be output. Packets may get re-ordered after
they enter the network. Suppose FMS receives two packets, 
$p_2$, $p_1$, which were re-ordered by the network before entering it. FMS itself does
not change their order.
It may output $\langle \triangledown \rangle$ corresponding to $p_2$ and $\langle p_{\delta_{11}},
 p_{\delta_{12}} \rangle$, corresponding to $p_1$. In general, FMS may output
no packet, or one or more altered or unaltered packets. It is possible that
due to arrival of a packet, there is no state change in an NF. 
 $P_{\delta}$ 
is totally ordered by $\prec$. The output of NFs is a sequence of a sequence of packets. For simplicity, we denote it as a sequence of packets  

When a flow $fl$ is to be migrated, SMM informs the source and
destination NFs involved in the migration that the flow migration is going to start. $m(k, q_{i})$ denotes the message sent to the other NF $k$ involved
in the flow migration and $q_{i}$ denotes the state to be updated in the NF
that receives the message. An NF sends this during the start of migration to send
the initial state $q_0$, and subsequently if and only if and when a state change
occurs during flow migration.  Updation of the NF that receives the message is
done by $U(m, q_x) = q_{j}$ --- the NF receives $m$ in state $q_x$ and updates
it to state $q_{j}$. SMM informs both the source and the destination NFs when a
flow migration is complete, indicating \textit{end} of flow migration, so that
they can stop updating each others' states. If no packets are being sent to the
sNF after a point in time, SMM can instruct the dNF to stop
updating the state of the sNF, in which case no message ($\bot$) will be
sent. 

The sequence of packets that are received at the sNF in Fig.
\ref{LON-main}(c), $NF_{1a}$, is denoted by $P_1$ and the sequence of packets
that are received at the dNF, $NF_{1b}$, is denoted by $P_2$. SMM
may cause re-ordering of packets. Thus $P_1$ and $P_2$ may not be in timestamp
order. There may be packets in $P$ that are neither present in $P_1$ nor in $P_2$, as some packets in $P$ may be instructed
to be dropped by SMM (if L is not preserved). ${P_1
\oplus} P_2$ denotes a sequence of packets that are present in $P_1$ or $P_2$ and
ordered according to their timestamps.

Fig. \ref{LON-main}(a) depicts an NF instance, called $NF_{1}$, that accepts
the sequence of packets $P_1 \oplus P_2$ and emits a modified
sequence of packets $P_{\delta}$, while undergoing state transformations. Let
the \textit{sequence} of state transformations for this sequence of packets be
$Q_1$.  An NF that does not drop, buffer or re-order packets is called an
\textit{ideal} NF.  $NF_1$ in Fig.~\ref{LON-main}(a) and Fig.~\ref{LON-main}(b)
are ideal NFs.  The sequence of states in $NF_{1}$, denoted by $Q_{NF_{1}}$,
follows the total order $\prec$ of the sequence of packets $P_1
\oplus P_2$ or $P$, as the case may be. The reason for two different inputs,
that is,  $P_1 \oplus P_2$ and $P$, is explained further ahead.
 
Let $Q_1$ be the finite sequence of
states created on $NF_{1a}$ and $Q_2$ on $NF_{1b}$. SMM migrates $Q_1^*$, a subset of $Q_1$, from
$NF_{1a}$ to $NF_{1b}$ and $Q_2^*$, a subset of $Q_2$, from $NF_{1b}$ to
$NF_{1a}$
during flow migration, depending on the flow
migration algorithm employed. It may instruct AM to buffer or drop packets or
both until the flow migration is complete.  A subset of $Q_2$ is considered, as
in some cases, only the latest state may need to be migrated to the destination
NF whereas in some others, after a state is migrated, newer packets may cause
state changes at the sNF, requiring further state migrations.

$P_{1\delta}^*$ denotes the sequence of packets output by $NF_{1a}$ and
$P_{2\delta}^*$ denotes the sequence of packets output by $NF_{1b}$. 
In general, a packet $p_i$, $1 \leq i \leq n$ is transformed to $s_i$
by $NF_{1a}$ or by $NF_{1b}$.  However, it is likely that some packets were
dropped by AM before reaching $NF_{1a}$ or $NF_{1b}$, due to the state
migration algorithm used.  Hence some packets $p_i$, $1 \leq i \leq n$, may not
be transformed to $s_i$ by $NF_{1a}$ or by $NF_{1b}$. $P_{\delta}^* =
P_{1\delta}^* \boxplus P_{2\delta}^*$ denotes the sequence of packets output by FMS
. The order of packets in $P_{\delta}^*$ is the order in which packets
in $P_{1\delta}^*$ and $P_{2\delta}^*$ arrive at the output of FMS, 
represented symbolically by a queue in Fig \ref{LON-main}(c). This may not
be in the order of packets in $P$.  In Fig \ref{LON-main}(a), the input to
$NF_1$ is the union of $P_1$ and $P_2$, the \textit{actual} sequence of packets
processed by either $NF_{1a}$ or $NF_{1b}$, but in the order of their
timestamps (this may not contain all the packets in $P$, as SMM may choose to
drop some packets).  Fig. \ref{LON-main}(b) represents an ideal NF whose input is $P$.
$P$ is the set of all packets that may enter an FMS
. The output then is $\overline{P_{\delta}}$.




\begin{definition}[Order preservation (O)]
\label{order-pres}
Let $Q_1$ be the finite sequence of states created on $NF_{1a}$
and  $Q_2$ on $NF_{1b}$ between the start and end of flow
migration. Let $Q_{NF_1}$ be the finite sequence of states created on the ideal
NF due to $P_1 \oplus P_2$ (Fig. \ref{LON-main}(a)). Migration of a flow $fl$ from an NF instance
$NF_{1a}$ to $NF_{1b}$ is Order-preserving iff 
1)  $Q_1$ is a prefix of $Q_{NF_1}$ and 2)  $Q_2$ is a suffix of $Q_{NF_1}$.
\end{definition}
\begin{remark}
Condition 1 (Condition 2) above ensures that all state updates in the source (destination) NF
are the same as those in the ideal NF. 
The flow migration system \textit{processes} packets
of a flow in the same manner that the sNF would have processed
them in the absence of migration. 

\end{remark}

\begin{definition}[External-order-preservation (E)]
The migration of $P$ to $NF_{1b}$ is
External-order preserving iff: 
$P_{\delta}^*$ is a subsequence of $\overline{P_{\delta}}$.
\end{definition}
\begin{remark}
$P_{\delta}^*$ is a \textit{subsequence} of $\overline{P_{\delta}}$ in the definition of
External-order-preservation to account for the fact that packets may be dropped during
flow migration. Also, the fact that packets may arrive out of order at $NF_{1b}$ needs
to be accounted for. That is why $\overline{P_{\delta}}$ is considered instead of 
$P_{\delta}$.  

\end{remark}

\begin{definition}[Loss-freedom(L)]
Let $\widehat{P_{\delta}^*}$ represent a set containing the elements 
of $P_{\delta}^*$, and  $\widehat{\overline{P_{\delta}}}$ represent a set
containing the elements of $\overline{P_{\delta}}$. 
The migration of $P$ to $NF_{1b}$ is Loss-free iff:
$\widehat{P_{\delta}^*} = \widehat{\overline{P_{\delta}}}$.
\end{definition}

\begin{definition}[No-buffering(N)]
The migration of $P$ to $NF_{1b}$ is
No-buffering preserving iff: no packet $p_i \in P$ is buffered by Action Manager.
\end{definition}

\begin{definition}[Strict Order preservation (SO)]
To Definition \ref{order-pres}, the following condition is added: 3) Both $P_1$ and $P_2$ are subsequences of $P$.
\end{definition}
\begin{remark}
\textit{In addition to preserving O},
to preserve SO, Condition 3 ensures that the flow migration system does not re-order packets as
they \textit{enter} the source or destination NFs.
''
\end{remark}
Messages are sent from one NF instance to another using a reliable transport mechanism so that they are not lost. 
\section{Correctness Criteria for flow migration}
In this section, we first discuss the correctness criteria for a single NF, 
assuming that there is no NF subsequent to the NF under discussion. This
is because the correctness criteria for flow migration for a given NF
will be affected if there are subsequent NFs, as described in the examples
in Section \ref{intro}.

\textbf{Preserving SO and E:} 
Even if there is no NF subsequent to the current NF,
networks expect that packets are not re-ordered and therefore, it is required
that E is preserved. This guarantees that packets are output in the same manner
as an ideal NF. In order to achieve this, a possible solution is that packets
are buffered at AM, and enough time is given to allow all packets to
traverse the source instance and exit the FMS. Then, state migration may be
initiated to the destination instance and after it is complete, the buffered
packets may be released to the destination. However, this requires
violating N. 
Preserving E, by definition, also ensures that
packets are processed in the same manner as an ideal NF, thus guaranteeing that
SO and O are preserved \cite{sukapuram2021loss}.
However, preserving O or SO requires immediate synchronization of every state updated at the
sNF with the dNF and vice versa. This requires incoming
packets to be buffered or dropped by the NF instances, until this state
synchronization is complete.
\begin{definition}[Correctness criteria for migrating a flow from one NF instance
to another]
Migration of a flow $fl$ from an NF instance $NF_{1a}$ to another instance $NF_{1b}$
is considered correct if the migration preserves the property of E.
\end{definition}
\newtheorem{theo}{Theorem}
\newtheorem{coro}{Corollary}[theo]

Consider an FMS that preserves SO. Consider two packets, $p_1$ and $p_2$, with
$p_1$ reaching FMS earlier than $p_2$. Suppose both do not cause state changes.
Let $p_1$ reach the sNF and $p_2$ the dNF. They reach these
NFs in the order of their time stamps. However, $p_2$ may exit the destination
NF before $p_1$, thus violating E, even though SO is preserved. In order to
preserve E, these packets need to be buffered at AM. Only after all packets
exit the sNF, which can be ensured by waiting for a suitable time, must
$p_2$ be sent to the dNF, to ensure that E is preserved.  Thus
preserving E requires packet buffering , which causes packet latency 
and increased storage requirements depending on the duration of migration.  Since
this is expensive, we examine if there is any property weaker than E.

\textbf{Relaxing O:}
In order to preserve SO or even O, either packets must be buffered or dropped, as per the LON theorem \cite{sukapuram2021loss}. Therefore,
we need to examine if O, SO and E can be relaxed  while avoiding
buffering or dropping packets. 
O requiring $Q_1$ to be a prefix of $Q_{NF_1}$ and $Q_2$ to be a suffix of $Q_{NF_1}$ is a stringent requirement, as it requires real-time synchronization of all states.  
If these conditions are relaxed for O, it may be possible to meet N and L.

\textit{The key insight
is that it is possible to consider a given state of a given NF
as a sequence of substates and each of those substates may belong to two types: the ones that need
updates in a synchronous manner and the ones that do not}. Flows may also be
partitioned into subsequences --- those that effect updates to the states that
require immediate synchronisation and those that do not.  For example, in a
NAT, when a SYN is received, the tuple of $\langle source\:IP\:address, source\:port \rangle$ of the
received packet is changed to one retrieved from a pool of such tuples. The
retrieved $\langle source\:IP\:address, source\:port \rangle$ is stored against the received tuple of
$\langle source\:IP\:address, source\:port \rangle$. To forward further packets received at the NAT
in either direction, knowledge of these entries, that form a substate, is
required.  Therefore this substate due to the first packet of a flow must be
synchronized with the dNF instance before the dNF
instance can forward any packet belonging to the flow.  The remaining substates
due to this packet and states due to the rest of the packets (for example,
timestamp of the latest received packet) 
 may be synchronized within a time bound, as they are not essential to
forward packets that are received in their absence.

We rely on prior investigations with regard to which NFs update states that are
relevant to a flow, on a per-packet basis~\cite{sadok2018case}. In prior work,
in order to make efficient use of CPU cores, 
packets are sent to different cores.
Since state modification is mostly done by a few packets at the beginning or at
the end of a connection, such packets are sent to the same core as the original
packet of the flow (a SYN) and the remaining packets are distributed across
cores.  It is observed that of the NFs surveyed, only a Deep Packet Inspection
(DPI) device requires writing state on a per-packet basis and other NFs such as
NATs, Firewalls etc. restrict writing to only a few packets per
flow~\cite{sadok2018case}.  Our paper is not limited to NFs that modify
states only at the beginning or at the end of a flow, though the majority
of NFs are in that category (A discusson on DPIs is omitted for brevity). Programmable data planes and NFs that implement
network functionalities may have more complex requirements. Therefore,
our discussion is, in general, for stateful nodes that maintain per-flow states. 

\textbf{Partitioning flows and states:}
\label{partition}
A state $q_i \in Q$ may be viewed as a sequence of $e$ \textit{substates}, denoted
as $q_{i1}, q_{i2},...q_{ie}$. It is possible that only a subsequence of $q_i$,
denoted as $S_i$, need to be synchronized with the dNF such that
packets arriving at the dNF can be forwarded. Upon arrival of a
sequence of packets of a flow, a sequence of states gets updated. In Fig.
\ref{substate}, corresponding to a packet arrival, the state $q_1$ gets updated.
$q_1$ consists of substates $\langle q_{11}, q_{12},q_{13},q_{14} \rangle$. $q_2$ consists of substates $\langle q_{21}, q_{22},q_{23},q_{24} \rangle$ and so on. 
Fig.~\ref{substate} is used as a running example for all definitions.

\begin{definition}[Partial equivalence of states]
\label{def:partial}
Two states $q_b$ and $q_c$ are partially equivalent with respect to a
subsequence of substates $Q_{ss}$, denoted as $\approxeq_{SS}$ if the values of
the substates corresponding to $Q_{ss}$ are equal in both.

Let $I$ denote the set of indices of the substates of $Q_{ss}$.
$\forall j \in I,  q_{bj} = q_{cj} \Rightarrow q_b \approxeq_{SS} q_c$.
For example, in Fig.
\ref{substate}, let $I=\{1,2\}$ and let $q_{11}= q_{21}$ and $q_{12}= q_{22}$. Then $q_1 \approxeq_{SS} q_2$.
\end{definition}
\begin{definition}[Partial equivalence of the outputs of the transformation function $f$]
Let $f(p_{b}, q_{b-1}) = \langle s_b, q_{b}, m(k_1,q_{b}) \rangle$
and 
 $f(p_{b}, q_{c-1}) = \langle s_c, q_{c}, m(k_1,q_{c}) \rangle$, for 
a given NF instance. $p_b$ is an arbitrary packet. When $f()$ is applied to $p_b$ and $q_{b-1}$ (alternatively,  $q_{c-1}$), $p_b$ is transformed to $s_b$ (alternatively, $s_c$), its state is changed to $q_{b}$ (alternatively, $q_c$) and a message $m$ is sent to the NF instance $k_1$. 
Let $Q_{ss}$ be
a subsequence of substates in $q_{b-1}$. Then
$(s_b = s_c) \land (q_b \approxeq_{SS} q_c) \Rightarrow 
f(p_{b}, q_{b-1})  \approxeq_{SS} f(p_{b}, q_{c-1})$.

The outputs of $f$ are partially equivalent with respect to a subsequence of
substates $Q_{ss}$ if the sequences of packets in the outputs are the same and if their
next states are partially equivalent with respect to $Q_{ss}$.
\end{definition}
\begin{definition}[Immediate Synchronisation]
\label{is}
Let us assume that
$p_xP_{cs}p_y$ is a sequence of packets in timestamp order of a flow $fl$ that is being migrated
from $NF_{1a}$ to another instance $NF_{1b}$ such that 
1) $p_x$ is a packet received at $NF_{1a}$ and $q_x$ the state due to
it and $q_x$ is instantaneously synchronized with $NF_{1b}$ ($p_x$ could
be null, in which case $q_x=q_0$)
2) $P_{cs}$ is a contiguous sequence of packets received at $NF_{1a}$ after $p_x$
3) $p_y$ is a packet received after $P_{cs}$
4) $Q_{ss}$
is a subsequence of substates that gets updated by packets of $P_{cs}$, at $NF_{1a}$
and $S_{css}$ is a sequence of these subsequences due to $P_{cs}$
5) $q_{y-1}$ is the state
at $NF_{1a}$ after it has processed $P_{cs}$. 

Let $S_{css}$ be a sequence of $Q_{ss}$ 
such that:
\label{Imm} $f(p_y, q_{y-1}) \not\approxeq_{SS} f(p_y, q_x)$.
In that case, $S_{css}$ is said to require
immediate synchronization, denoted by $S_{css} \models I$. $I$ denotes the set of indices of the substates of $Q_{ss}$, as explained in Definition \ref{def:partial}. $P_{cs}$ is said to effect the state changes in $S_{css}$,
denoted by $P_{cs} \rightarrow S_{css}$.
\end{definition}
\begin{remark}
This implies that for every packet $p$ in $P_{cs}$, corresponding to every
substate in the subsequence of substates $Q_{ss}$, a message $m(k,q_{ij})$
needs to be sent to update another instance $NF_{1b}$ of the
same stateful node, for any packet $p_y$ received after $P_{cs}$ at $NF_{1b}$
to be processed. $p_y$, received after $P_{cs}$ in the above definition
may belong to a reverse flow in a TCP connection, while $p_xP_{cs}$ may
belong to the forward flow.

Let $p_1p_2p_3p_4p_5$ be the packets of a flow. In Fig.
\ref{substate}, Let $q_1$ be the state at $NF_{1a}$ due to $p_1$. $q_1$ is instantaneously synchronized with $NF_{1b}$. Let $P_{cs}=p_2p_3p_4$ and let $p_y=p_5$. Let the indices $I$ corresponding to $Q_{ss}$ be $\{1,2\}$. Let  
$S_{css} =\{ \langle q_{21}, q_{22}\rangle, \langle q_{31}, q_{32} \rangle, \langle q_{41}, q_{42}\rangle\}$. After processing $P_{cs}$, let the state $q_{y-1}$ at $NF_{1a}$ be $q_4$. 
Let $f(p_5, q_{4}) \not\approxeq_{SS} f(p_5, q_1)$. In that case, $S_{css} \models I$. Thus the substates caused by $p_2p_3p_4$ on $NF_{1a}$ need to be immediately synchronized with $NF_{1b}$.

\end{remark}
\begin{figure}
\centerline{\includegraphics [scale=0.45] {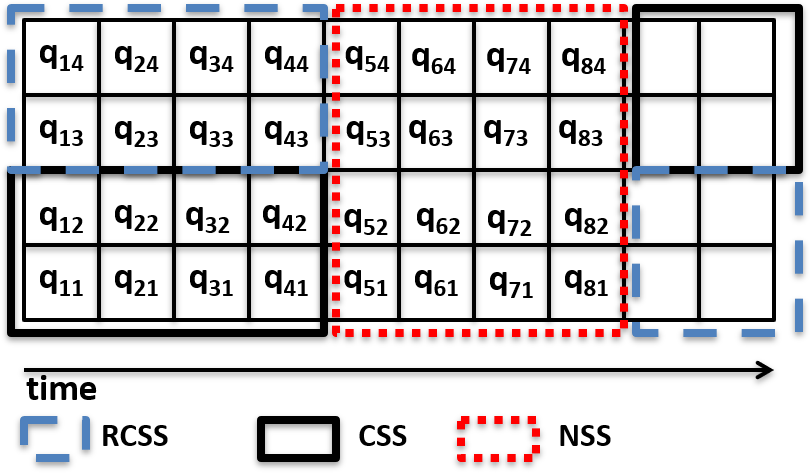}}
\caption{Sequences of states and substates}
\label{substate}
\end{figure}

\begin{definition}[Cohesive Subsequence]
\label{cs}
Let $P_{cs}^{\prime}$  be a sequence of packets
of a flow $fl$ undergoing migration. Let
$P_{cs}$ be the shortest subsequence of $P_{cs}^{\prime}$ such that
1) ($P_{cs} \rightarrow S_{css}) \land (S_{css} \models I$) and
2)
after $S_{css}$ is synchronized, $(P_{cs}^{\prime \prime} \rightarrow S_{css}^{\prime \prime}) \land (S_{css}^{\prime \prime} \not \models I)$ 
where $P_{cs}^{\prime \prime}$ is the subsequence after $P_{cs}$ 
is removed from $P_{cs}^{\prime}$
(after $S_{css}$ is synchronized, the remaining subsequence
$S_{css}^{\prime \prime}$ does not require immediate synchronization) 
then $P_{cs}$ is a Cohesive Subsequence of Packets (CSP).

In other words, the minimal contiguous subsequence of packets $P_{cs}$ that effect a
sequence of subsequence of substates $S_{css}$ that require immediate synchronization is called a
CSP. The sequence of subsequence of substates $S_{css}$ is called a
Cohesive Subsequence of Substates (CSS). The sequence of subsequence of substates due to a CSP that do not
need to be immediately synchronized is called the Remaining Cohesive
Subsequence of Substates (RCSS).

\end{definition}
The source address and port mappings of a NAT may constitute a CSS.
As another example, in Fig. \ref{substate}, the sequence of substates $\langle \{q_{11}, q_{12}\}, \{q_{21}, q_{22} \}, \{q_{31}, q_{32}\}, \{q_{41}, q_{42}\} \rangle $may form a CSS.  Then  the sequence of substates $\langle \{q_{13}, q_{14}\}, \{q_{23},
q_{24} \}, \{q_{33}, q_{34}\}, \{q_{43}, q_{44}\} \rangle$ form an RCSS. This may be a counter for the number of packets received so far and the maximum size of a packet respectively.
\begin{remark} In the definitions \ref{is} and \ref{cs}, the subsequence $P_{cs}$ may belong to a forward flow or a reverse flow or there may be one
each for both the forward and reverse flows.
\end{remark} 

\begin{definition}[Non-Cohesive Subsequence]
A contiguous subsequence of packets other than a CSP is called a Non-Cohesive
Subsequence of Packets (NSP). The corresponding sequence of states, if any,
is called a Non-Cohesive Subsequence of States (NSS).

\end{definition}
In the example in Fig. \ref{substate}, the sequence of states $\{q_{5},
q_{6}, q_7, q_8\}$ form an NSS. That is, these states do not need to be immediately synchronized.
 In a given flow, there may be multiple CSPs and NSPs. A CSP of a flow for a
given stateful node may not be a CSP for another stateful node. Similarly, an
NSP of a flow for a given stateful node may not be an NSP for another stateful node.


\begin{definition}[Eventual  Synchronization]
\label{eventual-sync}
During the duration of migration of a flow $fl$ from an NF instance $NF_{1a}$ to another NF instance $NF_{1b}$,  let $Q_2$ be the sequence of all states on $NF_{1b}$ after completion of flow migration and $Q_{NF_1}$ the sequence of states on an ideal NF due to $P_1 \oplus P_2$. If at least one suffix of $Q_2$ is a suffix of $Q_{NF_1}$, we say that the migration preserves Eventual  Synchronization of states.

\end{definition}

\begin{remark}
During state migration, the source and destination NFs may send messages
to each other to synchronize their RCSS and NSS states. The updates due
to these states must be in the order of the relative timestamps that
AM affixes to each packet sent to these NF instances. 
\textit{This paper assumes that updates to RCSS and NSS are commutative.}
\end{remark}

\begin{definition}[Weak-Order]
\label{wo}
During the duration of migration of a flow $fl$ from an NF instance $NF_{1a}$ to another NF instance $NF_{1b}$, 
1) R1: if all CSS of a
stateful node for $fl$, if any, are Immediately Synchronized with $NF_{1b}$ and
2) R2: if all states are eventually  synchronized,
we say that the property of Weak-Order is preserved.
\end{definition}

\begin{remark}
Weak-O does not prevent packets after a CSP and belonging to an NSP being processed by the sNF. In other words, Weak-O does not specify at what point packets may be forwarded to the dNF. This is specified by the algorithm that performs state migration.
\end{remark}

\section{Preserving Weak-O with N and L}
\label{section:WeakO}
As per the LON theorem, \cite{sukapuram2021loss}, it cannot be guaranteed that all three of L, O and N 
may be preserved simultaneously by a Flow Migration System  where no messages or packets are lost. We posit that
the property of Weak-O is useful because it provides a method for flow
migration that satisfies No-buffering, satisfies
Loss-freedom, 
and preserves packet order ``as much as required'', while
updating all states in timestamp order. The order in which packets
are sent to  NF instances are in the same order as they arrive at the FMS as long
as these packets effect a CSS; for packets that do not effect any CSS,
their order of arrival at  NF instances is not guaranteed. The packets corresponding
to a CSP are output from FMS in timestamp order while there is no guarantee about the
order in which packets corresponding to an NSP is output, whether they effect
state or not. Thus O (and E) are compromised to the extent that it remains useful,
to achieve L and N.

\begin{algorithm}
\caption{Migrate flows preserving Weak-O}
\label{alg-migrate}
\begin{algorithmic}[1]
\LineComment{Procedure for the sNF}
\Procedure {processpkts\_source}{}
\LineComment{Let $fl$ be a flow that is migrated from  $NF_{1a}$ to $NF_{1b}$. $p$ is a packet of $fl$. $Q$ is a queue to store the CSSs to be migrated and $Q^{\prime}$ a queue to store the RCSSs and NSSs to be migrated. States are timestamped.}
\State{$migration \;=\;TRUE$}
\State{$timer\_started\;=\;FALSE$}
\While {$migration$ is TRUE} \label{algo1:start}
		\State{Add CSS not sent to $NF_{1b}$, if any, to $Q$} \label{algo1:addq}
		\State{Send contents of $Q$ to $NF_{1b}$  }
		\State{Add RCSS and NSS not sent to $NF_{1b}$, if any, to $Q^{\prime}$} \label{algo1:addqrem}
		
	\If {$Q$ is empty \textbf{AND} end of current CSP (if any) is reached \textbf{AND} $timer\_started\;=\;FALSE$} \Comment{End of migration} 	\label{re-route}	
			\State{Inform the controller to re-route $fl$ such that further packets are sent to $NF_{1b}$}	
			\State{Start a timer for $T_r$ units}
			\Comment{$T_r$ is the maximum time to re-route a flow. This is to account for packets already on their way to $NF_{1a}$}
			\State{$timer\_started\;=\;TRUE$}
	\EndIf
        \State{Send contents of $Q^{\prime}$ to $NF_{1b}$}  \label{algo1:send}
	\If {$T_r$ expires \textbf{AND} $Q^{\prime}$ is empty} 
	\State{Break out of while loop}
	 \EndIf
\EndWhile
\EndProcedure

\end{algorithmic}
\end{algorithm}

\textbf{Algorithm that preserves Weak-O:}
\label{algo1}
Algorithm \ref{alg-migrate} illustrates the algorithm for the sNF.
After flow migration begins, CSS, if any, is added to a queue $Q$ (line \ref{algo1:addq}). RCSS and NSS, if any, are added to a queue $Q^{\prime}$ (line \ref{algo1:addqrem}). These continue to be sent from the sNF $NF_{1a}$ to the dNF $NF_{1b}$.  When the end of current CSS is reached and $Q$ is empty, further packets will be sent to $NF_{1b}$ by appropriately re-routing the flow (line \ref{re-route}). Since $NF_{1a}$ needs to process the packets that are already on the path towards it, a timer $T_r$ is started. $T_r$ is the maximum time required to re-route a flow (that is, the maximum time to send the first packet to $NF_{1b}$ after sending the CSS from $NF_{1a}$ to $NF_{1b}$). States stored in $Q^{\prime}$ continue to be sent to $NF_{1b}$. After all states have been sent to $NF_{1b}$ and the flow is re-routed (after $T_r$ expires), flow migration is complete. The dNF updates states as per messages received from the sNF, as described in Section \ref{LON-main}.

In Fig. \ref{prop}, if $NF_{1a}$ and $NF_{1b}$ are instances
of NATs, all packets of a flow are sent through $NF_{1a}$ until
$state[x_1,y_1]$ is migrated to $NF_{1b}$. Once migration completes, 
the remaining flow is sent through $NF_{1b}$, while migrating the remaining
substates and states, if any, to $NF_{1b}$. After the migration of $state[x_1,y_1]$, packets in transit towards $NF_{1a}$ may
modify RCSS and NSS, which do not require immediate synchronization to preserve Weak-O. After all state migration
is completed and packets are diverted to $NF_{1b}$, the flow migration
is completed, preserving Weak-O.

\begin{theo}
\label{theorem:preserves}
Algorithm \ref{alg-migrate} preserves Weak-O without buffering or dropping packets. 
\end{theo}

\begin{proof}
 Consider a flow $fl$ which is being migrated. Consider each packet $p$ of $fl$ as $p$ traverses the Flow Migration System while Algorithm \ref{alg-migrate} is being applied to it. Two invariants need to be preserved while processing each packet: $p$ must not be buffered or dropped and requirement R1 of the definition of Weak-Order  (Definition \ref{wo}, section \ref{partition}) needs to be met. We prove this first and prove that requirement R2 of the definition of Weak-Order is preserved, later.

Suppose $p$ is the first packet that is sent to the dNF $NF_{1b}$. We exhaustively enumerate the cases possible:
\begin{enumerate}
\item \label{1} $p$ does not belong to a CSP and there were no CSPs in the flow before $p$.
\item \label{2} $p$ does not belong to a CSP. CSS that came into effect before the arrival of $p$, if any, has been migrated already from the sNF to the dNF.
\item \label{3} $p$ does not belong to a CSP. CSS that came into effect before the arrival of $p$, if any, has not been migrated from the sNF to the dNF.
\item \label{4} $p$ belongs to a CSP. CSS that came into effect before the arrival of $p$, if any,  has been migrated. $p$ is the next packet. There is no packet belonging to the CSP before $p$ that has not updated state at the sNF.
\item \label{5}  $p$ belongs to a CSP. CSS that came into effect before the arrival of $p$, if any, has not been migrated. $p$ is the next packet. There is no packet belonging to the CSP before $p$ that has not updated state at the sNF.
\end{enumerate}

For cases \ref{1}, \ref{2}, and \ref{4}, the dNF can process $p$ by definition of CSP and CSS without dropping or buffering $p$. Moreover, $p$ has been forwarded to the dNF only after satisfying the conditions in line \ref{re-route} of Algorithm \ref{alg-migrate}. Therefore the invariants are preserved.

For cases \ref{3} and \ref{5}, $p$ will need to be buffered by definition of CSP and CSS  and therefore the invariants are not maintained. However, Algorithm \ref{alg-migrate} will not allow  these cases to occur, since as per line \ref{re-route} of Algorithm \ref{alg-migrate}, $p$ will be sent to $NF_{1b}$ only after CSS has been migrated. 

Therefore, as far as each packet sent to the dNF after the first is considered, the CSP, if any, has already been migrated. If the CSP has not occurred yet in the flow, a CSS will be created at the dNF and does not need to be migrated from the sNF. Additionally, as per line \ref{re-route}, once a packet is sent to the dNF, further packets are sent only there. Therefore for every packet $p$ in $fl$ the invariants are maintained.

Algorithm \ref{alg-migrate} will not allow cases \ref{3} and \ref{5} to occur, as established previously. Assume that for cases \ref{1}, \ref{2} and \ref{4} above, requirement  R2 of the defintion of Weak-Order (\ref{wo}, section \ref{partition}) does not hold good. Assume that a sequence of $n$ packets after $p$, that is, $p_1$,$p_2$,...,$p_n$ reach $NF_{1b}$ until migration ends and that a large enough time elapses after the packet before $p$ has traversed $NF_{1a}$. The sequence $p$, $p_1$,...,$p_n$ causes a sequence of states on $NF_{1b}$ and they are updated with messages from $NF_{1a}$ for RCSSs and NSSs due to packets earlier than $p$. Assume that the sequence of states at $NF_{1b}$  is $Q_2$. Let the suffix of the least length of $Q_2$ be $q$. Assume that $q$ is also not a suffix of $Q_{NF_1}$  (recall that $Q_{NF_1}$ is the sequence of states on an ideal NF due to $P_1 \oplus P_2$), as all states are not eventually synchronized and requirement R2 (\ref{wo}, section \ref{partition}) does not hold good. Then there are three possibilities: a) packets were dropped in the FMS b) messages from $NF_{1a}$ to $NF_{1b}$ for state updation were lost or c) updation of RCSS and NSS is not commutative. Since all these possibilities are not valid, the assumption is wrong and there exists a suffix of  $Q_2$ which is a suffix of $Q_{NF_1}$. Thus the algorithm for flow migration meets the requirement R2.

\end{proof}

\section{Discussion} 
\subsection{Conditions for preserving Weak-O when a flow is migrated:}
\label{conditions}
Let $I$ be the duration of a given NSP, $T_m$ the time to migrate a CSS and $T_r$ the
maximum time to inform an upstream switch to divert packets to the dNF. A migration
where packets do not need to be buffered is possible iff 
1) An NSP exists within $T_s$ units of time from the point at which migration
is desired, as a flow to be migrated must be migrated as early as possible.  $T_s$
is the latest time duration by which the start of flow migration can be delayed.
2) $T_m + T_r < I$. That is, the duration of such an NSS is large enough that
state migration can occur during the NSS interval and the packets that enter
the network after state migration
can be diverted to the dNF. The segregation of states into CSS, RCSS and NSS is assumed to be known for each NF.
\subsection{Implementation of the same NF in two ways}
\label{discussion}
A DPI device such as a Network Intrusion Detection Systems (NIDS) requires
matching a set of patterns that may occur in incoming traffic and if one of
these patterns match, the designated action for that pattern needs to be
performed. Since the pattern may occur across packet payloads, if packets are
received out-of-order, many systems buffer and re-order them before processing
them.  Suppose the desired action is to note the source IP address of the
packets matching a pattern.  Some packets carrying a part of the pattern may
arrive at the sNF and some other packets carrying a part of the same
pattern may arrive at the dNF. In this case, their states will
require Immediate Synchronisation, as information regarding all the packets
received so far is required before they are matched for a pattern.  Moreover,
in the worst case, multiple patterns may need to be matched and the occurrence
of all the patterns may not be across the same set of packets.  For such NF
implementations, in the worst case, no NSP exists and therefore no NSS.
Therefore, SO or O cannot be relaxed (Condition 1 of Section \ref{conditions} is not
satisfied).

Attackers deliberately re-order packets to cause NIDS to run out of buffer
space, thus making networks vulnerable to attacks. Recent
efforts \cite{yu20163}  search for patterns in re-ordered packets by saving
the prefixes and suffixes of arrived packets separately and matching them with
arriving packets, thus avoiding buffering packets and reducing the amount of
memory required.  When such an NF implementation is used, each packet may be
forwarded as soon as it is processed, without buffering, thus not requiring
Immediate Synchronisation and states may be Eventually 
Synchronized.  After synchronisation, the desired actions may be taken
depending on the conditions that are satisfied. For such NF implementations,
CSPs do not exist.  Therefore Weak-O is inexpensively implemented, that is,
without buffering packets.

\section{Implementation and Evaluation}
\begin{figure}[t]
\centering
\includegraphics [scale=0.7, trim={0 400 50 0},clip] {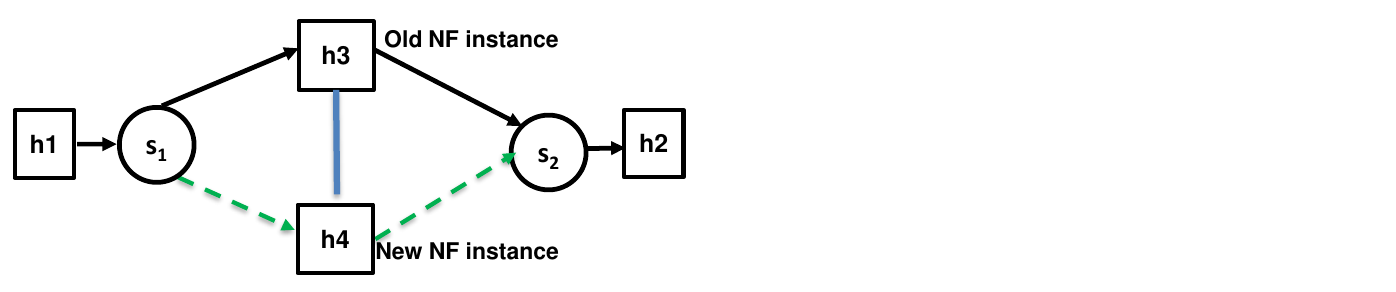}
\caption{Flow Migration from one NAT instance to another}
\label{fig:FlowMig}
\end{figure}
We implement a basic NAT 
and migrate a flow from one NAT instance to another. We use Mininet~\cite{mininet} (v 2.3.0) to implement a simple topology, shown in Fig. \ref{fig:FlowMig}.  $h_1$ and $h_2$ are hosts and $h_3$ and $h_4$ are instances of a NAT. The switches $s_1$ and $s_2$ are connected to a controller (not shown).  The switches are programmed such that flows from $h_1$ to $h_2$ traverse the path $h_1\!-\!s_1\!-\!h_3\!-\!s_2\!-\!h_2$. When a flow from $h_1$ to $h_2$ is to be migrated, the controller triggers migration of states from $h_3$ (\textit{old NAT}) to $h_4$ (\textit{new NAT}) after a time $t_1$ (Table \ref{table:symbols}) after the switches are up. The  states (mapping of (source IP address, source port) to (public IP address, public port)) are migrated as per Algorithm \ref{alg-migrate}. When migration is complete, $h_4$ informs the controller, which then updates the switch rules such that the flow is migrated to the path $h_1-s_1-h_4-s_2-h_2$. The updates are per-packet consistent  \cite{sukapuram2019ppcu}. 
We conduct experiments using this as the basic procedure. The system parameters and their default values are shown in Table \ref{table:symbols}.
The amount of data transferred is in multiples of $C$. 

A flow is started at $h_1$ and without state or flow migration, it is observed  by examining logs captured in Wireshark (version 4.0.6) that there is no packet loss or re-ordering on the old path.\textit{To verify that there is no packet loss or reordering on the new path,} a flow is started at $h_1$ and one packet of the flow is sent over the old path. The next packet is sent only after states have migrated and a time sufficient enough for the packet on the old path to have reached $h_2$ has elapsed. Now the rest of the packets of the flow are sent from $h_1$ over the new path. It is observed through logs  that packets are received without losses and in order at $h_2$, over the new path. 

The NAT implementation has a thread per port that reads packets from raw sockets and queues them. Another thread reads packets from this queue and processes them. The setup supports bandwidths upto $15$ Mbps per link without losses.

The objective of our experiments is to illustrate migration of state and flows without buffering and loss of packets but with some packet reordering, and to measure the effect of these on the goodput. The experiments are tested on a laptop (12th Gen Intel(R) 
Core (TM) i5-1235U, 1300 MhZ with 16 GB RAM) on a Linux VM with the default configuration for Mininet.  $20$ measurements are taken for all experiments. Wherever applicable, error bars are shown for a confidence interval of $95$\%.
The Mininet VM is reset before each set of measurements.  
\begin{table}
\caption{System parameters}
\begin{center}
\begin{tabular}{p{0.1\linewidth}  p{0.5\linewidth}p{0.25\linewidth}}
\toprule
\textbf{Symbol} & \textbf{Meaning} &\textbf{Default value} \\
\midrule
$t_1$ & Start time of flow migration & $15$s \\
$C$ & Data transferred in bytes & $1481481$\\
$B$ & Link bandwidth & 10 Mbps \\
$D$ & Link delay & 10 ms \\
$S$ & Payload size in the application$^*$ & 1000 bytes \\
\bottomrule
\end{tabular}
\label{table:symbols}
\end{center}
\footnotesize{$^*$The actual payload size will vary since TCP\_NODELAY is not set to True}
\end{table} 

\textbf{On Linux TCP:}
Before the last packet $p_o$ traversing the old path reaches $h_2$, the first packet that traverses the new path and which has a sequence number later than that of $p_o$, reaches $h_2$. Then $h_2$ sends a Duplicate Ack to $h_1$. In traditional TCP, when 3 Duplicate ACKs (\texttt{dupThresh} = 3) are received, the sender retransmits the packet and does not wait for the retransmission timeout to expire. This is called \textit{fast retransmit}. However, if some of those frames are on the way to the receiver but are only reordered and not lost, fast retransmits are not required. They decrease the goodput and reduce the congestion window size. In the version of Linux  (Ubuntu 20.04.1 LTS) used by Mininet, the TCP variant is Reno Cubic. Here Selective ACKs (SACKs) \cite{mathis1996rfc2018}  are sent by the receiver to indicate missing frames. These are used by the sender to find the extent of reordering and \texttt{dupThresh} is now varied depending upon the extent of reordering estimated, reducing false fast retransmits. 

Suppose there are false retransmits.  A Duplicate SACK (D-SACK) \cite{floyd2000rfc2883} is used to report a duplicate contiguous sequence of data received at the receiver. Once a D-SACK is received, the sender goes back to the old congestion window size that was being used before entering fast retransmit, thus attempting to nullify the effects of fast retransmit. Further discussions on the variant of TCP used in Linux are available in the literature \cite{arianfar2012tcp,johannessen2015investigate,ha2008cubic}.

We observe in the logs at $h_1$ that corresponding to each retransmit of a data packet, there is a D-SACK sent by the receiver, indicating that the retransmit was spurious. This avoids reduction of the congestion window and maintains the goodput. However, if the number of Duplicate Acks and retransmissions (even if they are later identified as spurious) is very large, goodput is affected as they take time to transmit.

\begin{figure*}[t]
\begin{subfigure}
{0.33\textwidth}
\includegraphics [scale=0.45] {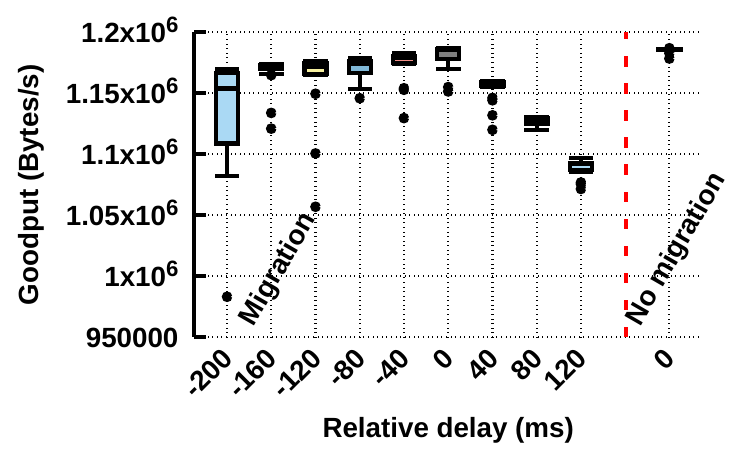}
\caption{Exp1: Goodput with and without migration for a flow of C*250 bytes for various relative delays of the old path.}
\label{fig:exp1-1}
\end{subfigure}
\begin{subfigure}
{0.33\textwidth}
\centerline{\includegraphics [scale=0.45] {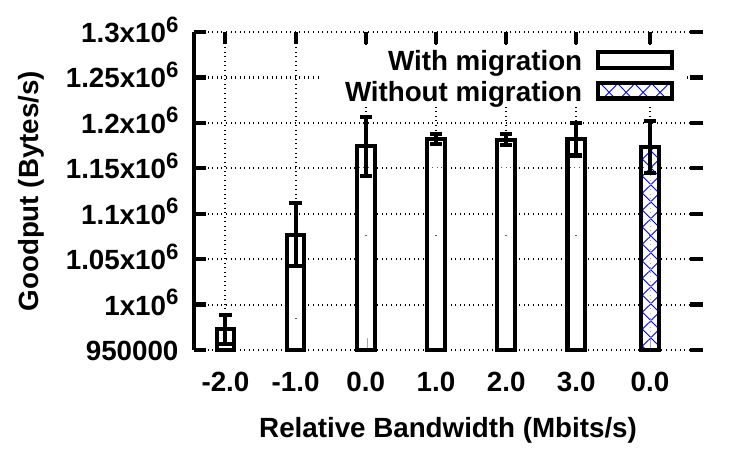}}
\caption{Exp4: Goodput with and without migration for a flow of C*200 bytes when link bandwidths of the new path are changed. 
}
\label{fig:exp6-1}
\end{subfigure}
\begin{subfigure}
{0.33\textwidth}
\centerline{\includegraphics [scale=0.45] {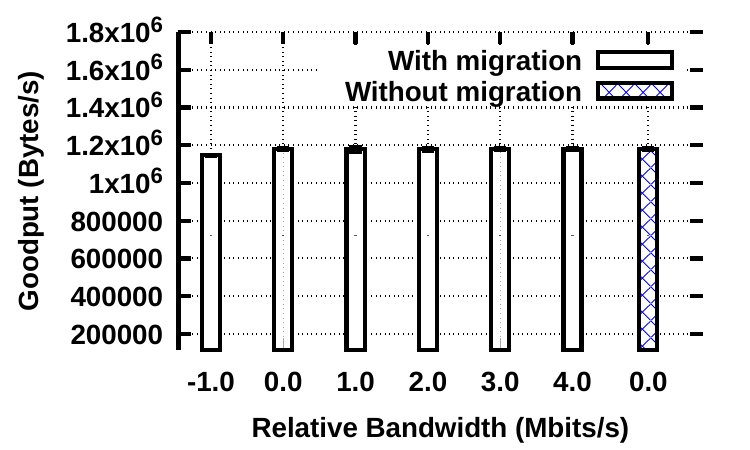}}
\caption{Exp.4: Goodput with and without migration for a flow of C*150 bytes when link bandwidths of the new path are changed. $t_1=15$s.}
\label{fig:exp6-2}
\end{subfigure}
\end{figure*}

\begin{figure*}[t]
\begin{subfigure}{0.33\textwidth}
\includegraphics[scale=0.48] {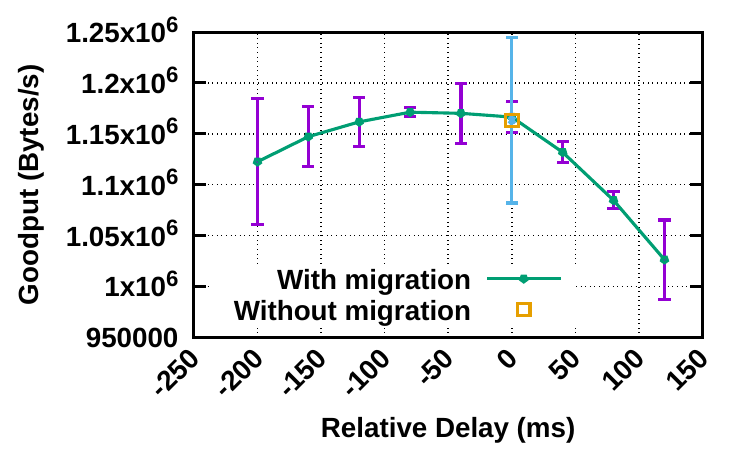}
\caption{For a flow of C*150 bytes transferred}
\label{fig:exp1-2-1}
\end{subfigure}
\begin{subfigure}{0.33\textwidth}
\includegraphics[scale=0.48] {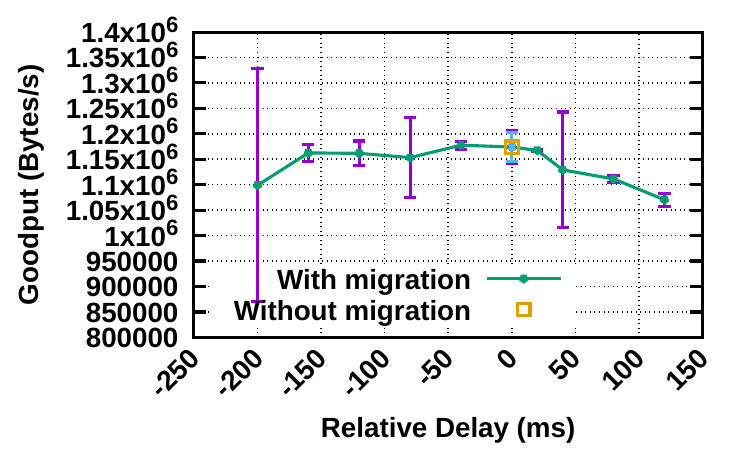}
\caption{For a flow of C*200 bytes transferred}
\label{fig:exp1-2-2}
\end{subfigure}
\begin{subfigure}{0.33\textwidth} \includegraphics[scale=0.48]
{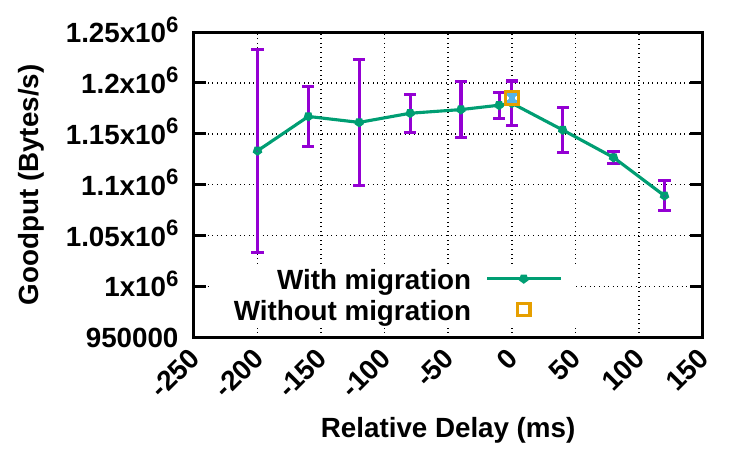}
\caption{For a flow of C*250 bytes transferred}
\label{fig:exp1-2-3}
\end{subfigure}
\caption{Exp1: Goodputs with and without migration for various relative delay of old path, for three values of bytes transferred}
\label{fig:exp1-2}
\end{figure*}
\textbf{Experiment 1:} We simulate the situation where the old and the new paths are asymmetric or have varied latencies.  We fix the size of the data transmitted to be $35.4$ MiB. ($250$ times C). We first increase the delay of the links on the old path with respect to the new path. We set delay differences of $0$, $40$, etc. per link and measure the goodput, when states are migrated after $t_1$ s from the start of the flow. The delays of all the other paths are always set to $10$ ms for this experiment. We also run the experiment by increasing the delays on only the new path, with respect to the old path (shown as negative delays in Fig.~\ref{fig:exp1-1}). The goodput measured for various differences in delays with respect to the old path are shown in Fig.~\ref{fig:exp1-1} as a box plot. For comparison, the maximum RTT, obtained from Wireshark  at $h_1$ in the case where there is migration of a flow is $320$ms and the average RTT is $134.9$ ms. The maximum of the RTT values is higher than the sum of the delays of the paths ($40$ ms) because of the queueing and processing delays in the nodes, mostly in the NAT. 
We also measure the goodput when there is no migration, with all links having equal delays.

Fig.~\ref{fig:exp1-2} illustrates the average goodput for flows of 3 sizes for various relative delays of links of the old path with respect to that of the new path.
When the old and the new paths have no difference in delay, the mean goodput remains the same as when there is no migration, in all the three plots. This is because even though Duplicate Acks are sent, a large number of retransmissions do not occur : a sample log for no difference in delays of the old and new paths and for a flow size of $C*150$, $10$ Duplicate ACKs and $2$ retransmissions were observed. 
As expected, the average goodput decreases when the delay of the old path increases, as  Duplicate Acks (some of which are SACKs indicating reordering) are sent for more packets, resulting in more retransmissions from $h_1$. A sample log for a delay of 40 ms for the old path contains 129 Duplicate ACKs (54 of which are D-SACKs) and 54 retransmissions, which reduces the goodput considerably. 

When links of the \textit{new} path have delays with respect to the old paths, the deterioration in goodput is only marginal, as illustrated in Fig.~\ref{fig:exp1-2} (shown as negative delay differences). For instance, in Fig.~\ref{fig:exp1-2-3}, when the number of bytes transferred is $C*250$ bytes and for a delay of $40$ms for each link of (only) the new path with respect to each link of the old path, the deterioration is $0.57$\%, for a delay of $80$ms it is $0.98$\%, for $120$ms and $160$ms it is $1.1$\%. A log for an instance when $C\!*\!150$ and when the new path has a delay of $40$ms shows no Duplicate ACKs and no re-ordering. This is because new packets do not reach $h_2$ to trigger Duplicate Acks. In conclusion, this method of migrating states and flows is feasible, with reasonable delay differences between the old and the new paths. We also conclude that measuring the goodput is more meaningful than comparing the extent of reordering of packets \cite{rfc5236}, due to the reordering sensitive congestion control algorithms used in Linux.

\begin{figure*}[t]
\begin{subfigure}{0.33\textwidth} 
\centerline{\includegraphics [scale=0.45] {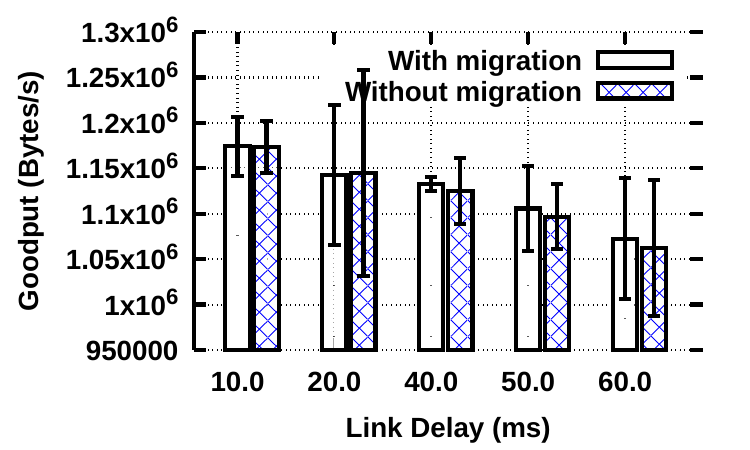}}
\caption{Goodput with and without migration for a flow of C*200 bytes when (absolute) link delays are changed}
\label{fig:exp2-1}
\end{subfigure}
\begin{subfigure}{0.33\textwidth}
\centerline{\includegraphics [scale=0.45] {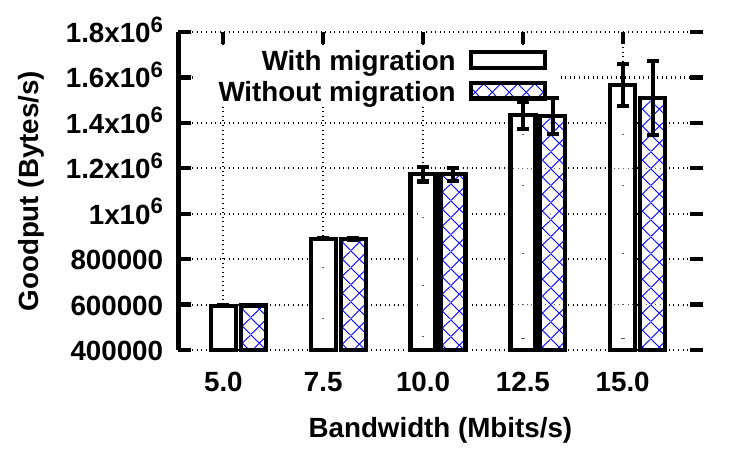}}
\caption{Goodput with and without migration for a flow of C*200 bytes when (absolute) link bandwidths are changed}
\label{fig:exp3-1}
\end{subfigure}
\begin{subfigure}
{0.33\textwidth}
\centerline{\includegraphics [scale=0.45] {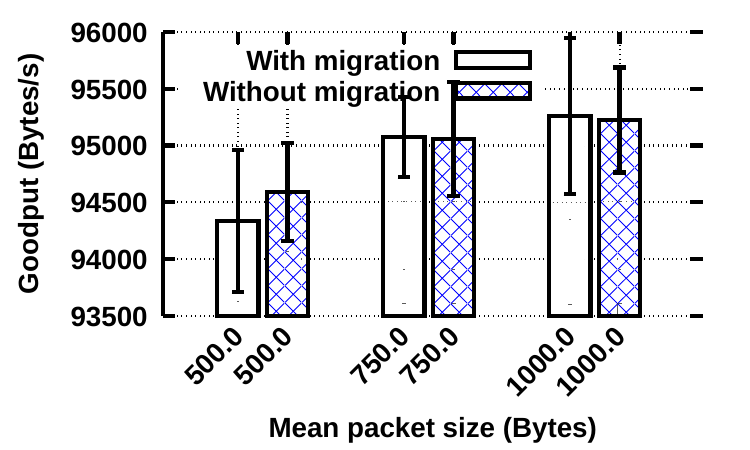}}
\caption{Goodput with and without migration for flows with realistic packet sizes and inter-arrival times}
\label{fig:exp7-1}
\end{subfigure}
\end{figure*}
\textbf {Experiments 2 and 3:} In experiment 2 (3) we perform Experiment 1 for various values of $D$ ($B$), with no relative delay between the old and the new paths, for a flow size of $C*200$ and for $B=10$Mbps ($D=10$ms) and measure the goodput. This is compared with a case where there is no migration. It is observed in Fig.\ref{fig:exp2-1} (Fig.\ref{fig:exp3-1}) that in all cases, there is no deterioration in the goodput when there is migration compared to the case when there is no migration. 

\textbf{Experiment 4:} In this experiment we change the bandwidths of the links of the new path relative to the old path and measure the goodput for various differences in bandwidth. As illustrated in Fig.\ref{fig:exp6-1}, increasing the bandwidth of the links of the new path causes the mean goodput to be comparable to that when there is no migration. Since the effects of migration are transient for a large flow ($C*200$) if the migration is started sufficiently early in the flow ($t_1=5$s considering an average flow duration of $25.3$s when there is no migration), we perform the same experiment for a short flow with $C*150$, with migration starting late in the flow, at $t_1=15$s (the average flow duration when there is no migration is $19.14$s). This is illustrated in Fig. \ref{fig:exp6-2}. 
Here too, the mean goodput is comparable to that when there is no flow migration.
The logs for a specific instance of the case when there is migration and when $B=14$ (Relative Bandwidth =4.0) for the new path (Fig. \ref{fig:exp6-2}) 
reveal the following: there are 71 Duplicate Acks out of which 20 are D-SACKs. 20 frames are retransmitted by $h_1$ and there are 20 D-SACKs, acknowledging that the retransmissions were spurious. 

\textbf{Experiment 5: }For SDNs, the inter-arrival time of packets is modelled as a Poisson process and the packet size distribution is also exponential \cite{lai2019performance}. Thus a TCP flow is generated with packet sizes and interarrival times following the above distribution and the difference in goodputs are plotted for various mean values of packet size and packet delay. If the packet size is $S$ bytes, the mean packet delay is $S$ microseconds. When $S=500$ , $t_1=5$s, when $S=750$, $t_1=10$s and when $S=1000$, $t_1=10$s, to ensure that migration  occurs when the flow is in progress. For all the cases tested, the goodputs with and without migration are comparable, as shown in Fig.\ref{fig:exp7-1}. When $S=500$, the difference in mean is only $0.27$\%.

\textbf{Conclusions:} 1. The time taken to migrate a flow does not matter (unless there is a requirement for that) if the flows are not buffered or packets are not dropped. There is no danger of high packet drops or buffers running out. 2. The number of packets that are reordered do not matter as long as they are detected as reordered and do not result in a large number of packet retransmissions. 
3. The goodput with and without migration are comparable when the old and new paths have the same delays and bandwidths, or when the new path has larger bandwidth or at most $5$ times longer delays.  
4. The above observations are valid for different values of absolute delays and absolute bandwidths of links. 5. Realistic differences in packet sizes and inter-arrival durations within a flow do not affect goodput when there is flow migration.

\section{Weak-O is the strongest criterion}
\label{section:strongest}

\begin{definition}[Strength of an implementation]
Let $Q_{1_{I_1}}$ and $Q_{2_{I_1}}$ be the sequence of states at the sNF and the corresponding dNF respectively between the start and stop of flow migration for an implementation $I_1$ of a flow migration algorithm, for an arbitrary flow. Let $Q_{1_{I_2}}$ and $Q_{2_{I_2}}$ respectively be the same for an implementation $I_2$. ($Q_{1_{I_1}}$, $Q_{2_{I_1}}$) is a \textit{valid} ordered pair if the sequence of states as given by this ordered pair can occur in some flow migration that satisfies the properties of $I_1$. Let $T_{I_1}$ be the set of all possible valid ordered pairs of  ($Q_{1_{I_1}}$, $Q_{2_{I_1}}$). Let $T_{I_2}$ be that of ($Q_{1_{I_2}}$,$Q_{2_{I_2}}$). 
An implementation $I_1$ of a flow migration algorithm is strictly stronger than another implementation $I_2$ if $T_{I_1} \subset T_{I_2}$. This is denoted by $I_1 \prec I_2$.
\end{definition}

Let $I_1$ be an implementation that does not preserve any property during migration. Let $I_w$ be an implementation that preserves Weak-O. Since all valid ordered pairs $(Q_{1_{I_w}}, Q_{2_{I_w}})$ of $I_w$ will be valid in $I_1$, but not vice-versa, $I_w \!\prec I_1$ holds.

\begin{theo}

 Let $I_w$ be an implementation of a  flow migration algorithm  that preserves Weak-O, C1)  that preserves L and N, and C2) with all states Eventually Synchronized. Let $I_1$ be an implementation of a  flow migration algorithm that meets conditions C1 and C2. Then $I_1 \not \prec I_w$. This assumes that flow migration is feasible.
\end{theo}

\begin{proof}

\textbf{Outline:}
Let us assume that $I_1 \prec I_w$ holds. 
Let $T_{I_1}$ and $T_{I_w}$ be the set of all possible ordered pairs of sequences of states in $I_1$ and $I_w$ respectively. The theorem states that $T_{I_1} \not\subset T_{I_w}$.
$T_{I_1}$ and $T_{I_w}$ are finite as it is assumed that state migration takes only a finite amount of time. So the number of states effected during that time is also finite.  



Consider a set of all possible ordered pairs of sequence of states $(Q_1,Q_2)\! \in\! T_{I_w}$ caused by migration of any flow $f$. Suppose there are migrations of $f$ that would result in the same set of ordered pairs of sequence of states $(Q_1,Q_2) \in T_{I_1}$.  Then  it would follow that for every $(Q_1,Q_2)\! \in\! T_{I_w}$, $(Q_1,Q_2)\!\in T_{I_1}$ will also hold. Thus $T_{I_w} \!\subseteq\! T_{I_1}$. Therefore it can be claimed that $T_{I_1}\! \not\!\subset\! T_{I_w}$ and therefore $I_1 \!\not \prec \!I_w$ would hold. This is the outline of the proof. We first consider all possible values of $Q_1$, followed by the corresponding values of $Q_2$. 
 
%

\textbf{On $Q_1$:} Consider a flow $fl$ undergoing state migration and the state that it effects, if any on the source and destination NFs. Let the sequence of states created by packets on the sNF \textit{until the first packet is sent to the dNF} be $Q_1$, for $I_w$. By virtue of the functionality of any NF, $Q_1$ is a valid first element of some ordered pair in $T_{I_1}$. 

Let a packet $p$ of a flow be the first packet that is sent to the dNF $NF_{1b}$, for $I_w$. We use the same exhaustive enumeration that was done in the proof of Algorithm \ref{alg-migrate}, which is reproduced here:
1) $p$ does not belong to a CSP and there were no CSPs in the flow before $p$.
2) $p$ does not belong to a CSP. CSS that came into effect before the arrival of $p$, if any, has been migrated already from the sNF to the dNF.
3) $p$ does not belong to a CSP. CSS that came into effect before the arrival of $p$, if any, has not been migrated from the sNF to the dNF.
4) $p$ belongs to a CSP. CSS that came into effect before the arrival of $p$, if any,  has been migrated. $p$ is the next packet. There is no packet belonging to the CSP before $p$ that has not updated state at the sNF.
5) $p$ belongs to a CSP. CSS that came into effect before the arrival of $p$, if any, has not been migrated. $p$ is the next packet. There is no packet belonging to the CSP before $p$ that has not updated state at the sNF.

For cases 3 and 5, $p$ will need to be buffered and therefore condition C1 is violated. Thus this sequence of states is invalid in both $T_{I_w}$ and $T_{I_1}$, as both require N to be preserved. However,
for cases 1, 2 and 4, the possible set of sequences of states is valid in both $T_{I_1}$ and $T_{I_w}$.  In these cases, CSS, if any, has been migrated to the dNF.

Let $p_1$ be the next packet in the flow. Assume that \textbf{A)} $p$ and $p_1$ belong to a CSP. \textbf{B)} $p$ belongs to a CSP but $p_1$ does not \textbf{C)} $p$ does not belong to a CSP but $p_1$ does \textbf{D)} Both $p$ and $p_1$ do not belong to a CSP. $p_1$ may be sent to \textbf{a)} the sNF or \textbf{b)} the dNF. Let us consider all possible combinations of the above cases. For cases \textbf{A-a} and  \textbf{B-a}, $p_1$ will need to be buffered. This is because if $p_1$ is sent to the sNF, since it belongs to a CSP, the state effected at the dNF by $p$ needs to be immediately synchronized with the sNF. Therefore these cases are invalid in both $I_1$ and $I_w$. \textbf{A-b}, \textbf{B-b}, \textbf{C-b}, and \textbf{D-b} are valid in both $I_1$ and $I_w$. The cases  \textbf{C-a} and \textbf{D-a}  remain to be discussed.

In $I_w$, \textbf{C-a} is not meaningful because the CSS effected due to $p_1$ at the sNF will need to be immediately synchronized with the dNF, and packets subsequent to $p_1$, unless sent to the sNF, will need to be buffered. Sending packets subsequent to $p_1$ to the sNF will delay migration, without any benefit. \textbf{D-a} is not efficient or meaningful in $I_w$ because after flow migration begins, if a packet that effects an NSS is sent to the sNF, the state will need to be eventually synchronized and if it does not effect a state, there is no reason why it should be sent to the sNF instance. \textit{Suppose \textbf{C-a} and \textbf{D-a}  are considered valid in $I_1$. Then it only adds state sequences to $T_{I_1}$ and removes no state sequences from $T_{I_1}$. No state sequences are removed from $I_1$ because implementations that do not allow \textbf{C-a} and \textbf{D-a}  are also valid in $T_{I_1}$}.  If we assume that all values of $Q_2$ are identical in $T_{I_1}$ and $T_{I_w}$ (we prove this next),  $T_{I_1} \not\subset T_{I_w}$. The above argument may be extended for packets subsequent to $p_1$.

\textbf{On $Q_2$:} Now let us consider the case after $fl$ is moved to $NF_{1b}$ and packets are no longer sent to $NF_{1a}$. 
Since $fl$ has moved to $NF_{1b}$, packets will no longer need to be buffered or dropped and thus C1 would never be violated.
Consider the sequence of states $Q_2$ created by the packets of $fl$ after and including $p$ on $NF_{1b}$. Note that a sequence of states $Q_2$ will correspond to some sequence of states $Q_1$ considered earlier. 

 At least one suffix of $Q_2$ is a suffix of $Q_{NF_1}$ (refer Table \ref{symbols}), as condition C2 needs to be satisfied, for both $I_1$ and $I_w$, by definition of Eventual Synchronization. Let $q_{I_w}$ be such a suffix for $I_w$ and $q_{I_1}$ that for $I_1$. Either $q_{I_w}= q_{I_1}$ or $q_{I_w}$ is a suffix of $q_{I_1}$  or $q_{I_1}$ is a suffix of $q_{I_w}$. In all the three cases, a suffix of least length $q_l$ may be found, that is a suffix of both $q_{I_1} $ and $q_{I_w}$. $q_l$ is also a suffix of $Q_{NF_1}$. Therefore any of the sequences of states  $Q_2$ is a valid second element of the ordered pair belonging to both $T_{I_1}$ and $T_{I_w}$. That is, there is no other condition to be considered. The first element of the ordered pair ($Q_1$,$Q_2$) is present  in $T_{I_1}$ due to packets until $p$ and has already been discussed.

Therefore, from the above two paragraphs, it can be concluded that for every $(Q_1,Q_2) \in T_{I_w}$, $(Q_1,Q_2) \in T_{I_1}$ will also hold and thus  $T_{I_w} \subseteq T_{I_1}$. Hence $T_{I_1} \not\subset T_{I_w}$.  Therefore $I_1 \not \prec I_w$ holds.

\end{proof}


\section{Related Work}
Our work builds on the LON theorem\cite{sukapuram2021loss} to explore correct flow migration and how relaxing O is useful. 

\textbf{Stateful Network Functions on servers:}
Migration of flows  and their network states from a Network Function instance to another is well-studied.
The properties of Loss-Freedom, Order and Strict Order are preserved in some of the work related to migrating flows (OpenNF ~\cite{gember2015opennf}, \cite{wang2018challenges}, SliM~\cite{nobach2017statelet}, ~SHarP\cite{peuster2018let}, \cite{wang2017consistent}, CHC~\cite{khalid2019correctness}, TFM~\cite{wang2016transparent}), but No-Buffering is not preserved. Another solution preserves only Loss-freedom while reducing the buffer size required ~\cite{gember2015improving}. Solutions that reduce buffer size required during migration do not eliminate the need for buffering altogether \cite{szalay2019industrial, khalid2019correctness, woo2018elastic}. These solutions take the approach of distributing states. However, since packets arriving at the dNF instance need to wait for state to be migrated from wherever it exists, No-buffering is not preserved. States may be stored outside NFs and fetched by the NF when required \cite{kablan2017stateless}. However, fetching the latest value of a state from this store will require packet buffering. Our solution eliminates buffering and packet drops  and reordering is managed within limits by the TCP variants on real networks. NFs that get affected by packet re-ordering are re-architected \cite{yu20163} to reduce vulnerabilities; this naturally supports flow migration preserving Weak-O.

\textbf{Stateful applications in Programmable Switches:} Since stateful processing is central to  applications such as load balancing 
, congestion control 
and packet scheduling, 
flow migration will require migrating state. Many such applications are briefly surveyed in \cite{ports2019should,xing2020secure,benson2019network}. Surveys on stateful data planes  \cite{hauser2023survey,ZHANG2021107597,kfoury2021exhaustive} have further information. Our solution is useful for all such applications.

\textbf{Other states:}
Maintaining consistency of  state updated across multiple flows  is dealt with elsewhere \cite{khalid2019correctness,wu2022nflow,muqaddas2020optimal}. While our paper deals only with per-flow states the algorithm proposed in this paper may be explored to update such states without buffering of packets.

\textbf{Packet re-ordering:}
Usb\"{u}t\"{u}n et al. argue that the variants of TCP used in Linux \cite{arianfar2012tcp,johannessen2015investigate,ha2008cubic} are robust to reordering and therefore reordering may not be an issue in the Internet \cite{usubutun2023switches}. Our work agrees with this observation.

\section{Conclusions}
In order to correctly migrate a flow from one
stateful NF instance to another, the property of External-Order
must be preserved.  Since this
requires packets to be buffered or dropped,  we proposed the
property of Weak-O, which requires real-time synchronization of state due to
only a (or a few) sequence(s) of packets, called the Cohesive Sequence of
Packets. 
This relaxation helps us eliminate packet buffering, thus reducing
packet latency and packet drops, while compromising External-Order. We
proposed an algorithm for flow migration that preserves Weak-O. 
 We implemented the algorithm and our evaluations show that the goodput with and without migration are comparable when the old and new paths have the same delays and bandwidths, or when the new path has larger bandwidth or at most $5$ times longer delays, for TCP flows. The widely used TCP variants on Linux that are robust to reordering makes the algorithm practical and  opens the possibility of better flow migration algorithms for stateful NFs. Moreover, we illustrated that NFs may be re-written to take advantage of Weak-O. We plan to investigate how to characterize this and how Weak-O may be preserved for migrations across NF chains.

\bibliographystyle{IEEEtran}
\bibliography{IEEEabrv,reference}

\end{document}